\documentclass[11pt]{article}
\usepackage[T1]{fontenc}
\usepackage[english]{babel}
\usepackage[latin1]{inputenc}
\usepackage{amsthm, amssymb,amsmath,eufrak,amsfonts}
\usepackage{graphicx}
\usepackage{dsfont}
\usepackage{colortbl}
\usepackage{color}

\newtheorem{theorem}{Theorem}[section]

\newtheorem{definition}{Definition}[section]
\newtheorem{example}{Example}[section]

\newtheorem{lemma}{Lemma}[section]

\newtheorem{proposition}{Proposition}[section]

\begin{document}

%%\tableofcontents
\date{}
\title{Axiomatizations for the Shapley-Shubik power index for games with several levels of approval in the input and output}

\author{Sascha Kurz\footnote{Department of Mathematics, Physics and
Computer Science, University of Bayreuth, 95440 Bayreuth, Germany
Tel.: +49-921-557353   Fax: +49-921-557352.  E--mail: sascha.kurz@uni-bayreuth.de}, \,
%\  {\normalsize and} \
Issofa Moyouwou\footnote{
Advanced Teachers Training College, University of Yaounde I, PO Box 47 Yaounde, Cameroon},\,
and Hilaire Touyem\footnote{Research and Training Unit for Doctorate in Mathematics, Computer Sciences and Applications, University of Yaounde I, PO Box 812, Yaounde, Cameroon}
}

\maketitle

\begin{abstract}
  The Shapley-Shubik index is a specialization of the Shapley value and is widely applied to evaluate the power distribution
  in committees drawing binary decisions. It was generalized to decisions with more than two levels of approval both in the input
  and the output. The corresponding games are called $(j,k)$~simple games. Here we present a new axiomatization for the
  Shapley-Shubik index for $(j,k)$~simple games as well as for a continuous variant, which may be considered as the limit case.

  \vskip 3mm

  \noindent {\it Math. Subj. Class. (2010)}: Primary 91A40,  91B12; Secondary 91A80, 91A12.

  \noindent {\it JEL Class.}: C71, D70, D71.
\end{abstract}

\section{Introduction}
\label{sec_introduction}

In \cite{shapley1953value} Shapley introduced a function that could be interpreted as the expected utility of a game from each of its positions via the
axiomatic approach -- the so-called Shapley value. A bit later, see \cite{shapley1954method}, it was restricted to games with binary decisions, i.e., simple games.
An axiomatization of this so-called Shapley-Shubik index was given quite a few years later by Dubey \cite{dubey1975uniqueness}. Nowadays, the Shapley-Shubik
index is one of the most established power indices for committees drawing binary decisions. However, not all decisions are binary. Abstaining from a vote might
be seen as a third option for the committee members. In general, there might also be any number $j\ge 2$ of alternatives that can be chosen from. To this end,
simple games were generalized to $(j,k)$ simple games \cite{freixas2003weighted}, where $j$ is the number of alternatives in the input, i.e., the voting possibilities,
and $k$ the number of alternatives for the group decision. A Shapley-Shubik index for these $(j,k)$ simple games was introduced in \cite{freixas2005shapley}
generalizing earlier attempts for special cases, see e.g.~\cite[pp. 291--293]{felsenthal1998measurement}. However, also other variants have been introduced
in the literature, see e.g.~\cite{amer1998extension,friedman2018conditional,hsiao1993shapley}. Here, we will only consider the variant from  \cite{freixas2005shapley}.
A corresponding axiomatizations is given in \cite{freixas2019value}.

If we normalize the input and output levels to numbers between zero and one, we can consider the limit if $j$ and $k$ tend to infinity for $(j,k)$ simple games.
More precisely we can consider the input levels $i/(j-1)$ for $0\le i\le j-1$ and the output levels $i/(k-1)$ for $0\le i\le k-1$. Then those games are discrete
approximations for games with input and output levels freely chosen from the real interval $[0,1]$. The later games were called simple aggregation functions in
\cite{kurz2018importance}, linking to the literature on aggregation functions \cite{grabischaggregation}, and interval simple games in \cite{kurz2019axiomatization}.
A Shapley-Shubik like index for those games was motivated and introduced in \cite{kurz2014measuring}, an axiomatization is given in \cite{kurz2019axiomatization}.

The success story of the Shapley-Shubik index for simple games, initiated by \cite{shapley1953value} and \cite{shapley1954method}, triggered a huge amount of
modifications and generalizations to different types of games, see e.g.\ \cite{algaba2019handbook} for some current research directions. We think that the
variants from \cite{freixas2005shapley}, for $(j,k)$-simple games, and from \cite{kurz2014measuring}, for interval simple games, form one consistent way
to generalize the Shapley-Shubik index for simple games. Here we mainly focus on an axiomatic justification, see our main result in Theorem~\ref{thm_charact_SSI_j_k}.  
Moreover, we present another formula for the Shapley-Shubik index for $(j,k)$ simple games which is better suited for computation issues, see Lemma~\ref{lemma_simplified_formula} 
and Theorem~\ref{thm_aver_form_SSI}. For a generalization of the Banzhaf index a similar result was obtained in \cite{pongou2012revenue}. 
%% and Theorem~\ref{Theo_Axiomatization-of SSI}.  
%% and specify the relation between $(j,k)$ simple games and interval simple games in greater detail.
As the title of the preface of \cite{algaba2019handbook} names it, the idea of the Shapley value is the root of
a still ongoing research agenda.

The remaining part of this paper is organized as follows. In Section~\ref{sec_preliminaries} we introduce the necessary preliminaries and present the first few
basic results. A Shapley-Shubik index $\Phi$ for general $(j,k)$ simple games is introduced in Section~\ref{sec_SSI_j_k}. Moreover, we study the first basic
properties of $\Phi$. In Section~\ref{sec_average_game} we introduce the average game, which is a TU game associated to each $(j,k)$ simple game. This notion
is then used to formulated the new axiom of average convexity, which culminates in an axiomatic characterization of $\Phi$ in Section~\ref{sec_axiomatization_j_k_1}.
In Section~\ref{sec_axiomatization_interval_simple_games} we transfer all notions and the axiomatic characterization to interval simple games.

\section{Preliminaries}
\label{sec_preliminaries}
Let $N=\left\{ 1,2,...,n\right\}$ be a finite set of voters. Any subset $S$ of $N$ is called a coalition and the set of all coalitions of $N$ is denoted by
the power set $2^{N}$. For given integers $j,k\ge 2$ we denote by $J=\{0,\dots,j-1\}$ the possible input levels and by $K=\{0,\dots, k-1\}$ the possible
output levels, respectively. We write $x\le y$ for $x,y\in\mathbb{R}^n$ if
$x_i\le y_i$ for all $1\le i\le n$. For each $\emptyset\subseteq S\subseteq N$ we write $x_S$ for the restriction
of $x\in\mathbb{R}^n$ to $\left(x_i\right)_{i\in S}$. As an abbreviation, we write $x_{-S}=x_{N\backslash S}$. Instead of $x_{\{i\}}$ and
$x_{-\{i\}}$ we write $x_i$ and $x_{-i}$, respectively. Slightly abusing notation we write $\mathbf{a}\in\mathbb{R}^n$, for the vector that
entirely consists of $a$'s, i.e., $\mathbf{0}$ for the zero vector.

A \emph{simple game} with player set $N$ is a mapping $v\colon 2^N\to \{0,1\}$ with $v(\emptyset)=1$, $v(N)=1$, and $v(S)\le v(T)$ for all
$\emptyset\subseteq S\subseteq T\subseteq N$. Coalitions $S\subseteq N$ with $v(S)=1$ are called \emph{winning} and \emph{losing} otherwise.
The interpretation in the voting context is as follows. Those elements $i\in N$, called \emph{voters} or \emph{players}, that are contained
in a coalition $S$ are those that are in favor of a certain proposal. The other voters, i.e., those in $N\backslash S$, are against the proposal.
If $v(S)=1$ then the proposal is implemented and otherwise the status quo persists. A simple
game $v$ is \emph{weighted} if there exists a \emph{quota} $q\in\mathbb{R}_{>0}$ and \emph{weights} $w_i\in\mathbb{R}_{\ge 0}$ for all $i\in N$ such
that $v(S)=1$ iff $w(S):=\sum_{i\in S} w_i\ge q$. As notation we use $\left[q;w_1,\dots w_n\right]$ for a weighted (simple) game. An example is
given by $v=[4;3,2,1,1]$ with winning coalitions $\{1,2\}$, $\{1,3\}$, $\{1,4\}$, $\{1,2,3\}$, $\{1,2,4\}$, $\{1,3,4\}$, $\{1,2,3,4\}$, and $\{2,3,4\}$.
A simple game $v$ is a \emph{unanimity game} if there exists a coalition $\emptyset\neq T\subseteq N$ such that $v(S)=1$ iff $T\subseteq S$. As
an abbreviation we use the notation $\gamma_T$ for a unanimity game with defining coalition $T$. It is well known
that each simple game admits a representation as disjunctions of a finite list of unanimity games. Calling a winning coalition \emph{minimal} if all
proper subsets are losing, such a list is given by the minimal winning coalitions, i.e., by $\{1,2\}$, $\{1,3\}$, $\{1,4\}$, and $\{2,3,4\}$ in the above
example.

If being part of a coalition is modeled as voting {\lq\lq}yes{\rq\rq} and {\lq\lq}no{\rq\rq} otherwise, represented as $1$ and $0$, respectively, then
one can easily reformulate and generalize the definition of a simple game:
\begin{definition}
  \label{def_j_k_simple_game}
  A \emph{$(j,k)$ simple game for $n$ players}, where $j,k\ge 2$ and $n\ge 1$ are integers, is a mapping $v\colon J^n\to K$ with $v(\mathbf{0})=0$,
  $v(\mathbf{j-1})=k-1$, and $v(x)\le v(y)$ for all $x,y\in J^n$ with $x\le y$. The set of all $(j, k)$ simple games on $N$ is denoted by $\mathcal{U}_n^{j,k}$
  or by $\mathcal{U}_n$, whenever $j$ and $k$ are clear from the context.
\end{definition}
So, $(2,2)$ simple games are in one-to-one correspondence to simple games.
We use the usual ordering of $J$ (and $K$) as a set of integers, i.e., $0<  1< \cdots< j-1$. In words, in the input set, $0$ is the lowest level
of approval, followed by $ 1$ and so on. In general, we call a function $f\colon\mathbb{R}^n\supseteq U\to \mathbb{R}$ \emph{monotone} if we have
$f(x)\le f(y)$ for all $x,y\in U$ with $x\le y$. We remark that in \cite{freixas2005shapley} the author considers a more general definition of a $(j,k)$
simple game than we have presented here. Additionally the $j$ input levels and the $k$ output levels are given by a so-called numeric evaluation.
Our case is called uniform numeric evaluation there, which motivated the notation $\mathcal{U}_n$ for $(j,k)$ simple games for $n$ players. We also
call a vector $x\in J^n$ a \emph{profile}.

\begin{definition}
  \label{def_null_player_j_k}
  Given a $(j,k)$ simple game $v$ with player set $N$, we call a player $i\in N$ a \emph{null} player if $v(x)=v(x_{-i},y_i)$ for all $x\in J^n$ and all
  $y_i\in J$. Two players $i,h\in N$ are called \emph{equivalent} if $v(x)=v(x')$ for all $x,x'\in J^n$ with $x_l=x'_l$ for all $l\in N\backslash \{i,h\}$,
  $x_i=x'_h$, and $x_h=x'_i$.
\end{definition}

In words, a player $i$ is a null player if its input $x_i$ does not alter the output $v(x)$. If interchanging the input $x_i$ and $x_h$ of two players does never
alter the output $v(x)$, then players $i$ and $h$ are equivalent. By $\pi_{ih}$ we denote the transposition on $N$ interchanging $i$ and $h$, so that the
previous condition reads $v(x)=v(\pi_{ih} x)$ for all $x\in J^n$.  By $\mathcal{S}_n$ we denote the set of permutations of length $n$, i.e., the bijections on $N$.

Now let us introduce a subclass of $(j,k)$ simple games with the property that for each profile $x$, the collective decision $v(x)$ is either $0$
\emph{(the lowest level of approval)} or it is $k-1$ \emph{(the highest level of approval)} depending on whether some given voters report some minimum approval
levels. For example, when any full support of the proposal necessitates a full support of each voter in a given coalition $S$, players in $S$ are each empowered
with a veto. One may require from each player in $S$ only a certain level of approval for a full support of the proposal. All such games will be called $(j,k)$
simple games with point-veto.

\begin{definition}
  \label{def_point_veto}
  A $(j,k)$ simple game with a \emph{point-veto} is a $(j,k)$ simple game $v$ such that there exists some $a\in J^n\backslash \{\mathbf{0}\}$ satisfying
  $v(x)=k-1$ if $a\leq x$ and $v(x)=0$ otherwise for all $x\in J^n$. In this case, $a$ is the veto and the game $v$ is denoted by $u^a$. For each coalition
  $S\in 2^N$ we abbreviate $w^S=u^a$, where $a_i=j-1$ for all $i\in S$ and $a_i=0$ otherwise.
\end{definition}

We remark that $(2,2)$ simple games with a point veto are in one-to-one correspondence to the subclass of unanimity games within simple games. The set of all players
who report a non-null approval level is denoted by $N^a$, i.e., $N^a=\left\{i\in N:0<a_i\leq j-1\right\}$. Every player in $N^a$ will be called a \emph{vetoer} of
the game $u^a$. Note that for the vector $a$ defined via $w^S=u^a$ we have $N^a=S$.

Null players as well as equivalent players can be identified easily in a given $(j,k)$ simple game with point-veto:
\begin{proposition}
  \label{prop_null_and_equivalent_players}
  Let $a\in J^n\backslash\{\mathbf{0}\}$. A player $i\in N$ is a null player of $u^a$ iff $i\in N\backslash N^a$.
  Two players $i,h\in N$ are equivalent in $u^a$ iff $a_i=a_h$.
\end{proposition}
\begin{proof}
  For every $a\in J^n\backslash\{\mathbf{0}\}$ and every $i\in N\backslash N^a$ we have $a_i=0$ by the definition of $N^a$. Now let $i\in N\backslash N^a$. For
  every $x\in J^n$ and every $y_i\in J$ we have $a\leq x$ iff $a\leq \left(x_{-i},\, y_i\right)$. Thus, $u^a(x)=u^a\left(x_{-i},\, y_i\right)$ and $i$ is a null player in $u^a$.
  Now let $i\in N^a$, i.e., $a_i>0$. Since $v(a)=k-1\neq 0=v\left(a_{-i},\mathbf{0}_i\right)$, player $i$ is not a null player in $u^a$.

  Assume that $a_i=a_h$ and consider an arbitrary $x\in J^n$. Then we have $a\leq x$ if and only if $a\leq \pi_{ih}x$. The definition of $u^a$ directly gives $u^a(x)=u^a(\pi_{ih}x)$,
  so that the players $i$ and $h$ are equivalent in $u^a$. Now suppose that the players $i$ and $h$ are equivalent in $u^a$. Since  $a\leq a$, we obtain $u^a(a)=u^a(\pi_{ih}a)=k-1$.
  This implies that $a\leq \pi_{ih}a$.  Therefore $a_i\leq a_h$ and $a_{h}\leq a_i$, that is $a_i=a_h$.
\end{proof}

Note that $(j,k)$ simple games can be combined using the disjunction $(\vee)$ or the conjunction $(\wedge)$ operations to obtain new games.
\begin{definition}
  \label{def_vee_and wedge}
  Let $v'$ and $v''$ be two $(j,k)$-simple games with player set $N$. By $v' \vee v''$ we denote the $(j,k)$ simple game $v$ defined by $v(x)=\max\{v'(x),v''(x)\}$ for all
  $x\in J^n$. Similarly, by $v'\wedge v''$ we denote the $(j,k)$ simple game $v$ defined by $v(x)=\min\{v'(x),v''(x)\}$ for all $x\in J^n$.
\end{definition}
We remark that the defining properties of a $(j,k)$ simple game can be easily checked. This can be specialized to the subclass of $(j,k)$ simple games with point veto,
i.e., $(j,k)$ simple games with point-veto can be combined using the disjunction $(\vee)$ or the conjunction $(\wedge)$ operations to obtain new games. To see this, consider a
non-empty subset $E$ of $J^n\backslash\{\mathbf{0}\}$ and define the $(j,k)$ simple game denoted by $u^E$ by $u^E(x)=k-1$ if $a\leq x$ for some $a\in E$ and $u^E(x)=0$ otherwise,
where $x\in J^n$ is arbitrary. Note that the notational simplification $u^{\{a\}}=u^a$, where $a\in J^n\backslash \{\mathbf{0}\}$, goes in line
with Definition~\ref{def_point_veto}.

\begin{proposition}
  \label{prop_vee_and wedge}
  Let $E$ and $E'$ be two non-empty subsets of $J^n\backslash\{\mathbf{0}\}$. Then, we have  $u^E \vee u^{E'}=u^{E\cup E'}$ and $u^E\wedge u^{E'}=u^{E''}$, where
  $E''=\{c\in J^n: c_i=max(a_i,b_i) \text{ for some } a\in E \text{ and } b\in E'\}$.
\end{proposition}
\begin{proof}
  In order to prove $u^E \vee u^{E'}=u^{E\cup E'}$ we consider an arbitrary $x\in J^n$. If $u^{E\cup E'}(x)=k-1$, then  there exists $a\in E\cup E'$ such that, $a\leq x$.
  Therefore $u^E(x)=k-1$ or $u^{E'}(x)=k-1$  and  $(u^E \vee u^{E'})(x)=k-1$. Now suppose that $u^{E\cup E'}(x)=0$. Then, for all $a\in E\cup E'$ we have $a\nleq x$. Since
  $E\subseteq E\cup E'$ and $E'\subseteq E\cup E'$ we have $b\nleq x$ and $c\nleq x$ for all $b\in E$ and all $c\in E'$. This implies that $u^E(x)=u^{E'}(x)=0$ and
  $(u^E \vee u^{E'})(x)=0$. Thus, $u^E \vee u^{E'}=u^{E\cup E'}$.

  Similarly, in order to prove $u^E\wedge u^{E'}=u^{E''}$ we consider an arbitrary $x\in J^n$. If $u^{E''}(x)=k-1$, then there exists $c\in E''$ such that $c\leq x$. But,
  by definition of $E''$, $c=\max (a, b)$ for some $a\in E$ and $b\in E'$, that is $a\leq c\leq x$ and $b\leq c\leq x$. Hence, $u^E(x)=u^{E'}(x)=k-1$ and
  $(u^E\wedge u^{E'})(x)=k-1$. Now assume that $u^{E''}(x)=0$ and $(u^E\wedge u^{E'})(x)\neq 0$. By definition of $u^E$ and $u^{E'}$, we have $(u^E\wedge u^{E'})(x)=k-1$.
  Thus, there exists $a\in E$ and $b\in E'$ such that $a\leq x$ and $b\leq x$. It follows that $c=\max (a, b)\leq x$, which is a contradiction to $u^{E''}(x)=0$. This proves
  $u^E\wedge u^{E'}=u^{E''}$.
\end{proof}

For $(j,k)=(5,3)$ and $n=3$ an example is given by $E=\left\{(1,2,3),(2,1,2)\right\}$, $E'=\left\{(4,1,1),(1,1,3)\right\}$. With this, $E''=\left\{(4,2,3),(1,2,3),(2,1,3),(4,1,2)\right\}$.
Note that we may remove $(4,2,3)$ from that list since $(4,2,3)\ge (1,2,3)$ (or $(4,2,3)\ge (4,1,2)$).

Especially, Proposition~\ref{prop_vee_and wedge} yields that every $(j,k)$ simple game of the form $u^E$ is a disjunction of some $(j,k)$ simple games with point-veto. So, each
$(j,k)$ simple game of the form
$u^E$ will be called a $(j,k)$ simple game with \emph{veto}. In the game $u^E$, $E$ can be viewed as some minimum requirements (or thresholds) on the approval levels of voters'
inputs for the full support of the proposal. It is worth noticing that $u^E$ is $\{0, k-1\}$-valued; the final decision at all profiles is either a no-support or a full-support.
The set of all veto $(j,k)$ simple games on $N$ is denoted $\mathcal{V}_n$. Note that Proposition~\ref{prop_vee_and wedge} shows that $\mathcal{V}_n$ is a lattice.

The sum of two $(j,k)$ simple games cannot be  a $(j,k)$ simple game itself. However, we will show that each $(j,k)$ simple game is a convex combination of $(j,k)$ simple games with veto.

\begin{definition}
  A convex combination of the games $v_1,v_2,\dots,v_p\in  \mathcal{U}_n$ is given by $v=\sum_{t=1}^p\alpha_tv_t$ for some non-negative numbers $\alpha_t$, where $t=1,2,\dots,p$,
  that sum to $1$.
\end{definition}
Note that not all convex combinations of $(j,k)$ simple games are $(j,k)$ simple games.

\begin{proposition}
  \label{prop_conv_decomp_j_k}
  For each $(j,k)$ simple game $v$ there exist  a collection of positive numbers $\alpha_t$, where $t=1,2,\dots,p$, that sum to $1$ and a collection $F_t(v)$, where $t=1,2,\dots,p$,
  of non-empty subsets of $J^n$ such that $v=\sum_{t=1}^p\alpha_tu^{F_t(v)}$.
\end{proposition}
\begin{proof}
  Let $v\in \mathcal{U}_n$ and $\mathcal{F}(v)=\{x\in J^n, v(x)> 0\}$. Since $J^n$ is finite and $v$ is monotone, the elements of $\mathcal{F}(v)$ can be labeled in such a way that
  $\mathcal{F}(v)=\{x^1, x^2,\dots, x^p\}$, where $x^p=\mathbf{1}$, $v(x^t)\leq v(x^{t+1})$ for all $1\leq t< p$, and $t\leq s$ whenever $x^t\leq x^s$. Now, set $x^0=\mathbf{0}$ and
  ${F}_t(v)=\{x^s,\, t\leq s\leq p\}$, $\alpha_t=\frac{v(x^t)-v(x^{t-1})}{k-1}$ for all $1\leq t\leq p$. By our assumption on $x^t$ we have $\alpha_t\geq0$ for all $1\leq t\leq p$. Moreover,
  it can be easily checked that $\sum_{t=1}^p \alpha_t=\frac{v(x^p)-v(x^0)}{k-1}=1$. set $u=\sum_{t=1}^p\alpha_tu^{F_t(v)}$.

  In order to prove that $v=u$, we consider an arbitrary $x\in J^n$. First suppose that $x\notin \mathcal{F}(v)$. Since $v$ is monotone, there is no $a\in \mathcal{F}(v)$ such that $a\leq x$.
  By definition, it follows that $v^{F_t(v)}(x)=0$ for all $t=1,2,\dots,p$. Therefore $v(x)=u(x)=0$. Now suppose that $x\in \mathcal{F}(v)$. Then $x=x^s$ for some $s=1,2,\dots,p$.  It follows
  that for all $t=1,2,\dots,p$ we have $v^{F_t(v)}(x)=k-1$ if $1\leq t\leq s$ and $v^{F_t(v)}(x)=0$ otherwise. Therefore
  $$u(x)=\sum_{t=1}^s\alpha_t=\sum_{t=1}^s\left(\frac{v(x^t)-v(x^{t-1})}{k-1}\cdot (k-1)\right)=v(x^s)=v(x).$$
  Clearly, the game $v$ is a convex combination of the games $u^{F_t(v)}$, where $t=1,2,\dots,p$.
\end{proof}

Proposition \ref{prop_conv_decomp_j_k} underlines the importance of $(j,k)$ simple games with veto, i.e., every $(j,k)$ simple game can be obtained from $(j,k)$ simple games with
veto as a convex combination.

\medskip

Now let us consider a continuous version of $(j,k)$ simple games normalized to the real interval $I:=[0,1]$ for the input as well as the output levels. Following
\cite{kurz2018importance} and using the name from \cite{kurz2019axiomatization}, we call a mapping $v\colon[0,1]^n\to[0,1]$ an \emph{interval simple game} if
$v(\mathbf{0})=0$, $v(\mathbf{1})=1$, and $v(x)\le v(y)$ for all $x,y\in [0,1]^n$ with $x\le y$. Replacing $J$ by $[0,1]$ in Definition~\ref{def_null_player_j_k}
we can transfer the concept of a null player and that of equivalent players to interval simple games.

\section{The Shapley-Shubik index for simple and $(j, k)$ simple games}
\label{sec_SSI_j_k}

Since in a typical simple game $v$ not all players are equivalent, the question of influence of a single player $i$ on the final group decision $v(S)$ arises.
Even if $v$ can be represented as a weighted game, i.e., $v=[q;w]$, the relative individual influence is not always reasonably reflected by
the weights $w_i$. This fact is well-known and triggered the invention of power indices, i.e., mappings from a simple game on
$n$ players to $\mathbb{R}^n$ reflecting the influence of a player on the final group decision. One of the most established power indices
is the \emph{Shapley-Shubik index} \cite{shapley1954method}. It can be defined via
\begin{equation}
  \label{eq_ssi_simple_games}
  \operatorname{SSI}_{i}(v)=\sum_{i\in S\subseteq N}\frac{(s-1)!(n-s)!}{n!}\cdot \left[
  v(S)-v(S\backslash \left\{ i\right\} )\right]
\end{equation}
for all players $i\in N$, where $s=|S|$. If $v(S)-v(S\backslash\{i\})=1$, then we have $v(S)=1$ and $v(S\backslash\{i\})=0$ in a simple game
and voter~$i$ is called a \emph{swing voter}.

In \cite{shapley1954method} the authors have motivated the Shapley-Shubik index by the following interpretation. Assume that
the $n$ voters row up in a line and declare to be part in the coalition of {\lq\lq}yes{\rq\rq}-voters. Given an ordering of the players,
the player that first guarantees that a proposal can be put through is then called \emph{pivotal}. Considering all $n!$ orderings $\pi\in\mathcal{S}_n$ of
the players with equal probability then gives a probability for being pivotal for a given player $i\in N$ that equals its
Shapley-Shubik index. So we can rewrite Equation~(\ref{eq_ssi_simple_games}) to
\begin{equation}
  \label{eq_roll_call_simple_games_first}
  \operatorname{SSI}_i(v)=\frac{1}{n!}\cdot \sum_{\pi\in\mathcal{S}_n} \Big( v(\{j\in N\,:\,\pi(j)\le \pi(i)\})-v(\{j\in N\,:\,\pi(j)< \pi(i)\}) \Big).
\end{equation}
Setting $S_\pi^i:=\{j\in N\,:\,\pi(j)\le \pi(i)\}$ we have $S_\pi^i=S$ for exactly $(s-1)!(n-s)!$ permutations $\pi\in\mathcal{S}_n$ and
an arbitrary set $\{i\}\subseteq S\subseteq N$, so that Equation~(\ref{eq_ssi_simple_games}) is just a simplification of
Equation~(\ref{eq_roll_call_simple_games_first}).

Instead of assuming that all players vote {\lq\lq}yes{\rq\rq} one can also assume that all players vote
{\lq\lq}no{\rq\rq}. In \cite{MannShapley} it is mentioned that the model also yields the same result if we assume
that all players independently vote {\lq\lq}yes{\rq\rq} with a fixed probability $p\in[0,1]$. This was further generalized to probability
measures $p$ on $\{0,1\}^n$ where vote vectors with the same number of {\lq\lq}yes{\rq\rq} votes have the same probability, see
\cite{hu2006asymmetric}. In other words, individual votes may be interdependent but must be exchangeable. That no further probability
measures lead to the Shapley-Shubik index was finally shown in \cite{kurz2018roll}. For the most symmetric case $p=\tfrac{1}{2}$ we can rewrite
Equation~(\ref{eq_roll_call_simple_games_first}) to
\begin{equation}
  \label{eq_roll_call_simple_games}
  \operatorname{SSI}_i(v)=\frac{1}{n!\cdot 2^n}\cdot \sum_{(\pi,x)\in\mathcal{S}_n\times\{0,1\}^n} M(v,(\pi,x),i),
\end{equation}
where $M(v,(\pi,x),i)$ is one if player $i$ is pivotal for ordering $\pi$ and vote vector $x$ in $v$, see \cite{kurz2018roll},
and zero otherwise.

This line of reasoning can be used to motivate a definition of a Shapley-Shubik index for $(j,k)$ simple games as defined in \cite{freixas2005shapley},
c.f.\ \cite{kurz2014measuring}. Suppose that voters successively and independently each choose a level of approval in $J$
with equal probability. Such a vote scenario is modeled by a roll-call $(\pi,x)$ that consists in a permutation $\pi$ of the voters and a profile $x\in J^n$ such for all $i\in N$,
the integer $\pi(i)\in \{1,2,\dots,n \}$ is the entry position of voter $i$ and $x_i$ is his approval level. Given an index $h\in\{1,\dots,k-1\}$, a voter $i$ is an $h$-pivotal voter
if the vote of player $i$, according to the ordering $\pi$ and the approval levels of his predecessors, pushes the outcome to at least $h$ or to at most $h-1$.

\begin{example}
  \label{ex_j_k_simple_game}
  Let $v$ be the $(3,3)$ simple game $v$ for $2$ players defined by $v(0, 0) =v(1,0)= 0$, $v(1, 1)=v(0, 1) = 1$, and
$v(2,0)=v(0,2)=v(2,1)=v(1,2)=v(2,2)=2$. As an example, consider the ordering $\pi=(2,1)$, i.e., player $2$ is first, and the vote vector $x=(2,1)$. Before player 2 announce his vote
$x_2=1$ all outcomes in $K=\{0,1,2\}$ are possible. After the announcement the outcome $0$ is impossible, since $v(0,1)=1$, while the outcomes $2$ and $3$ are still
possible. Thus, player $2$ is the $1$-pivotal voter. Finally, after the announcement of $x_1=2$, the outcome is determined to be $v(2,1)=2$, so that player $1$ is
the $2$-pivotal voter.
\end{example}

Going in line with the above motivation and the definition from \cite{freixas2005shapley}, the Shapley-Shubik index for $(j,k)$~simple games is defined for
all $v\in \mathcal{U}_n$ and for all $i\in N$ by:

\begin{equation}
  \label{eq_roll_call_jk_simple_games}
  \Phi_i(v)=\frac{1}{n!\cdot j^n\cdot (k-1)}\sum_{h=1}^{k-1} \left|\left\{(\pi,x)\in\mathcal{S}_n\times J^n\,:\,i\text{ is an $h$-pivot for
  $\pi$ and $x$ in $v$}\right\}\right|.
\end{equation}

Since several different definitions of a Shapley-Shubik index for $(j,k)$-simple games have been introduced in the literature, we prefer to
use the more inconspicuously notation $\Phi_i(v)$ instead of the more suggestive notation $\operatorname{SSI}_i(v)$. For the $(j,k)$ simple game $v$ from
Example~\ref{ex_j_k_simple_game} we have
$$
  \Phi(v)=\left(\Phi_1(v),\Phi_2(v)\right)=\left(\frac{5}{12}, \frac{7}{12}\right).
$$
Hereafter, some properties of $\Phi$ are explored. To achieve this, we introduce further definitions and axioms for power indices on $(j,k)$ simple games.
First of all, we simplify Equation~(\ref{eq_roll_call_jk_simple_games}) to a more handy formula.
\begin{lemma}
  \label{lemma_simplified_formula}
  For each $(j,k)$ simple game $v\in\mathcal{U}_n$ and each player $i\in N$ we have
  \begin{equation}
    \label{eq_roll_call_jk_simple_games_simplified}
    \Phi_i(v)=\sum_{i\in S\subseteq N} \frac{(s-1)!(n-s)!}{n!}\cdot\left[C(v,S)-C(v,S\backslash\{i\})\right],
  \end{equation}
  where $s=|S|$ and
  \begin{equation}
    \label{eq_def_C_function}
    C(v,T)=\frac{1}{j^n(k-1)}\cdot \sum_{x\in J^n} \Big(v(\mathbf{(j-1)}_T,x_{-T})-v(\mathbf{0}_T,x_{-T})\Big)
  \end{equation}
  for all $T\subseteq N$.
\end{lemma}
\begin{proof}
  For a given permutation $\pi\in\mathcal{S}_n$ and $i\in N$, we set $\pi_{<i}=
  \left\{j\in N\,:\, \pi(j)<\pi(i)\right\}$, $\pi_{\le i}= \left\{j\in N\,:\, \pi(j)\le\pi(i)\right\}$, $\pi_{>i}=\left\{j\in N\,:\,
  \pi(j)>\pi(i)\right\}$, and $\pi_{\ge i}= \left\{j\in N\,:\, \pi(j)\ge \pi(i)\right\}$. With this, we can rewrite
  $n!\cdot j^n\cdot (k-1)$ times the right hand side of
  Equation~(\ref{eq_roll_call_jk_simple_games}) to
  \begin{equation}
    \label{eq_uncertainty_reduction_jk_simple_games}
    %%\frac{1}{n!\!\cdot\! j^n\!\cdot\! (k-1)}
    \sum_{(\pi,x)\in\mathcal{S}_n\times J^n}\!\! \Big(\!\!\left[v(x_{\pi_{<i}},\mathbf{(j-1)}_{\pi_{\ge i}})-v(x_{\pi_{<i}},
    \mathbf{0}_{\pi_{\ge i}})\right]
    -\left[v(x_{\pi_{\le i}},\mathbf{(j-1)}_{\pi_{> i}})-v(x_{\pi_{\le i}},\mathbf{0}_{\pi_{> i}})\right]\!\!\Big).
  \end{equation}
  The interpretation is as follows. Since $v$ is monotone, before the vote of player~$i$ exactly the values in
  $\left\{v(x_{\pi_{<i}},\mathbf{0}_{\pi_{\ge i}}),\dots, v(x_{\pi_{<i}},\mathbf{(j-1)}_{\pi_{\ge i}})\right\}$ are still possible as final group decision. After
  the vote of player~$i$ this interval eventually shrinks to $\left\{v(x_{\pi_{\le i}},\mathbf{0}_{\pi_{> i}}),\dots, v(x_{\pi_{\le i}},\mathbf{(j-1)}_{\pi_{> i}})\right\}$.
  The difference in (\ref{eq_uncertainty_reduction_jk_simple_games}) just computes the difference between the lengths of both intervals, i.e.,
  the number of previously possible outputs that can be excluded for sure after the vote of player~$i$.

  As in the situation where we simplified the Shapley-Shubik index of a simple game given by Equation~(\ref{eq_roll_call_simple_games_first})
  to Equation~(\ref{eq_ssi_simple_games}), we observe that it is sufficient to know the sets $\pi_{\ge i}$ and $\pi_{>i}$ for every permutation
  $\pi\in\mathcal{S}_n$. So we can condense all permutations that lead to the same set and can simplify the expression in  (\ref{eq_uncertainty_reduction_jk_simple_games})
  and obtain Equation~(\ref{eq_roll_call_jk_simple_games_simplified}).
\end{proof}

While we think that the roll-call motivation stated above for Equation~(\ref{eq_roll_call_jk_simple_games}) is a valid justification on its own,
we also want to pursue the more rigor path to characterize power indices, i.e., we want to give an axiomatization. A set of properties that are satisfied
by the Shapley-Shubik index for simple games and uniquely characterize the index is given, e.g., in \cite{shapley1953value,shapley1954method}. In order
to obtain a similar result for $(j,k)$ simple games, we consider a \emph{power index} $F$ as a map form $v$ to $\mathbb{R}^n$ for all
$(j,k)$ simple games $v\in\mathcal{U}_n$.

\begin{definition}
  \label{def_power_index_properties}
  A power index $F$ for $(j,k)$ simple games satisfies
  \begin{itemize}
    \item \emph{Positivity} (P) if $F(v)\neq\mathbf{0}$ and $F_i(v)\ge 0$ for all $i\in N$ and all $v\in \mathcal{U}_n$;
    \item \emph{Anonymity} (A) if $F_{\pi(i)}(\pi v)=F_i(v)$ for all permutations $\pi$ of $N$, $i\in N$, and $v\in\mathcal{U}_n$, where
          $\pi v (x)=v(\pi(x))$ and $\pi(x)=\left(x_{\pi(i)}\right)_{i\in N}$;
    \item \emph{Symmetry} (S) if $F _{i}(v)=F_{j}(v)$ for all $v\in \mathcal{U}_{n}$ and all voters $i,j\in N$ that are equivalent in $v$;
    \item \emph{Efficiency} (E) if $\sum_{i\in N}F _{i}(v)=1$ for all $v\in \mathcal{U}_{n}$;
    \item the \emph{Null player property} (NP) if $F _{i}(v)=0$ for every null voter $i$ of an arbitrary game $v\in \mathcal{U}_n$;
    \item the \emph{transfer property (T)} if for all $u,v\in \mathcal{U}_n$ and all $i\in N$ we have $F_i(u)+F_i(v)=F_i(u\vee v)+F_i(u\wedge v)$,
          where $(u\vee v)(x)=\max\{u(x),v(x)\}$ and $(u\wedge v)(x)=\min\{u(x),v(x)\}$ for all $x\in J^n$, see Definition~\ref{def_vee_and wedge}
          and Proposition~\ref{prop_vee_and wedge};
    \item \emph{Convexity} (C)  if $F(w)=\alpha F(u)+ \beta F(v)$ for all $u, v\in \mathcal{U}_{n}$ and all $\alpha,\beta\in\mathbb{R}_{\ge 0}$ with
           $\alpha+\beta=1$,  where $w=\alpha u+\beta v\in \mathcal{U}_n$;
    \item \emph{Linearity} (L) if $F(w)=\alpha F(u)+ \beta F(v)$ for all $u, v\in \mathcal{U}_{n}$ and all $\alpha,\beta\in\mathbb{R}$,  where $w=\alpha u+\beta v\in \mathcal{U}_n$.
  \end{itemize}
\end{definition}
Note that $\alpha\cdot u+ \beta\cdot v$ does not need to be a $(j,k)$ simple game for $u,v\in\mathcal{U}_n$, where $\alpha\cdot u$ is defined via $(\alpha\cdot u)(x)=\alpha\cdot u(x)$
for all $x\in J^n$ and all $\alpha\in \mathbb{R}$. We remark that, obviously, (L) implies (C) and (L) implies (T). Also (S) is implied by (A).  Some of the properties
of Definition~\ref{def_power_index_properties} have been proven to be valid for $\Phi$ in \cite{freixas2005shapley}. However, for the convenience of the
reader we give an extended result and a full proof next:

\begin{proposition}
  \label{prop_power_index_properties}
  The power index $\Phi$, defined in Equation~(\ref{eq_roll_call_jk_simple_games}), satisfies the axioms (P), (A), (S), (E), (NP), (T), (C), and (L).
\end{proposition}
\newcommand{\upim}{x_{\pi_{<i}},\left(\mathbf{j-1}\right)_{\pi_{\ge i}}}
\newcommand{\upi}{x_{\pi_{\le i}},\left(\mathbf{j-1}\right)_{\pi_{>i}}}
\newcommand{\dnim}{x_{\pi_{<i}},\mathbf{0}_{\pi_{\ge i}}}
\newcommand{\dni}{x_{\pi_{\le i}},\mathbf{0}_{\pi_{>i}}}
\newcommand{\upimk}{x_{\kappa_{<i}},\left(\mathbf{j-1}\right)_{\kappa_{\ge i}}}
\newcommand{\upik}{x_{\kappa_{\le i}},\left(\mathbf{j-1}\right)_{\kappa_{>i}}}
\newcommand{\dnimk}{x_{\kappa_{<i}},\mathbf{0}_{\kappa_{\ge i}}}
\newcommand{\dnik}{x_{\kappa_{\le i}},\mathbf{0}_{\kappa_{>i}}}
\newcommand{\upimky}{y_{\kappa_{<j}},\left(\mathbf{j-1}\right)_{\kappa_{\ge j}}}
\newcommand{\upiky}{y_{\kappa_{\le j}},\left(\mathbf{j-1}\right)_{\kappa_{>j}}}
\newcommand{\dnimky}{y_{\kappa_{<j}},\mathbf{0}_{\kappa_{\ge j}}}
\newcommand{\dniky}{y_{\kappa_{\le j}},\mathbf{0}_{\kappa_{>j}}}
\newcommand{\upimkx}{x_{\kappa_{<j}},\left(\mathbf{j-1}\right)_{\kappa_{\ge j}}}
\newcommand{\upikx}{x_{\kappa_{\le j}},\left(\mathbf{j-1}\right)_{\kappa_{>j}}}
\newcommand{\dnimkx}{x_{\kappa_{<j}},\mathbf{0}_{\kappa_{\ge j}}}
\newcommand{\dnikx}{x_{\kappa_{\le j}},\mathbf{0}_{\kappa_{>j}}}
\newcommand{\upjm}{x_{\pi_{<j}},\left(\mathbf{j-1}\right)_{\pi_{\ge j}}}
\newcommand{\upj}{x_{\pi_{\le j}},\left(\mathbf{j-1}\right)_{\pi_{>j}}}
\newcommand{\dnjm}{x_{\pi_{<j}},\mathbf{0}_{\pi_{\ge j}}}
\newcommand{\dnj}{x_{\pi_{\le j}},\mathbf{0}_{\pi_{>j}}}
\begin{proof}
  We use the notation from the proof of Lemma~\ref{lemma_simplified_formula} and let $v$ be an arbitrary $(j,k)$
  simple game with $n$ players.

  For each $x\in J^n$, $\pi\in\mathcal{S}_n$, and $i\in N$, we have $v(\upim)\ge v(\upi)$
  and $v(\dni)\ge v(\dnim)$, so that $\Phi_i(v)\ge 0$ due to Equation~(\ref{eq_uncertainty_reduction_jk_simple_games}). Since we will show that $\Phi$ is efficient, we especially
  have $\Phi(v)\neq \mathbf{0}$, so that $\Phi$ is positive.

  For any permutation $\pi\in\mathcal{S}_n$ and any $0\le h\le n$ let $\pi|h:=\{\pi(i)\,:\,1\le i\le h\}$, i.e.,
  the first $h$ players in ordering $\pi$. Then, for any profile $x\in J^n$, we have
  \begin{eqnarray*}
    &&\sum_{i=1}^n \Big(v(\upim)-v(\upi)+v(\dni)-v(\dnim)\Big)\nonumber\\
    &=&\sum_{h=1}^n \Big(v(x_{\pi|h-1},\mathbf{(j-1)}_{-\pi|h-1})\!-\! v(x_{\pi|h},\mathbf{(j-1)}_{-\pi|h})\Big)
    \!+\!\sum_{h=1}^n \Big(v(x_{\pi|h},\mathbf{0}_{-\pi|h})\!-\!v(x_{\pi|h-1},\mathbf{0}_{-\pi|h-1})\Big)\nonumber\\
    &=& v(x_{\pi|0},\mathbf{(j-1)}_{-\pi|0})-v(x_{\pi|n},\mathbf{(j-1)}_{-\pi|n})+v(x_{\pi|n},\mathbf{0}_{-\pi|n})-
    v(x_{\pi|0},\mathbf{0}_{-\pi|0})\nonumber\\
    &=& v(\mathbf{(j-1)})-v(x)+v(x)-v(\mathbf{0})=k-1-0=k-1,
  \end{eqnarray*}
  so that Equation~(\ref{eq_uncertainty_reduction_jk_simple_games}) gives $\sum_{i=1}^n \Phi_i(v)=1$, i.e., $\Phi$ is efficient.

  The definition of $\Phi$ is obviously anonymous, so that it is also symmetric. If player $i\in N$ is a null player and
  $\pi\in\mathcal{S}_n$ arbitrary, then $v(\dnim)=v(\dni)$ and  $v(\upim)=v(\upi)$, so that $\Phi_i(v)=0$, i.e., $\Phi$
  satisfies the null player property. Since Equation~(\ref{eq_uncertainty_reduction_jk_simple_games}) is linear in the
  involved $(j,k)$ simple game, $\Phi$ satisfies (L) as well as (C), which is only a relaxation. Since $x+y=\max\{x,y\}+\min\{x,y\}$
  for all $x,y\in\mathbb{R}$, $\Phi$ also satisfies the transfer axiom (T).
\end{proof}

Actually the proof of Proposition~\ref{prop_power_index_properties} is valid for a larger class of power indices for $(j,k)$ simple games. To
this end we associate each vector $a\in J^n$ with the function %%the TU game
$v_a$ defined by
$$
  v_a(S)=\displaystyle\frac{1}{k-1}\cdot[v(\mathbf{(j-1)}_S, a_{-S})-v(\mathbf{0}_S, a_{-S})]
$$
for all $S\subseteq N$. With this, we define the mapping $\Phi^{a}$  on $\mathcal{U}_n$ by
\begin{equation}
  \label{eq_parametric_index}
  \Phi_i^{a}(v) =\sum_{i\in S\subseteq N}{\displaystyle\frac{(s-1)!(n-s)!}{n!}[v_{a}(S)-v_{a}(S\backslash \{i\}]}
\end{equation}
for all $i\in N$. We remark that it can be easily checked that $v_a$ is a TU game, c.f.\ Section~\ref{sec_average_game}.

Similar as in the proof of Lemma~\ref{prop_power_index_properties}, we conclude:
\begin{proposition}
  \label{prop_parametric_power_index}
  For every $a\in J^n$ such that $a_i=a_j$ for all $i,j\in N$, the mapping $\Phi^{a}$ the axioms (P), (A), (S), (E), (NP), (T), (C), and (L).
\end{proposition}

While the Shapley-Shubik index for simple games is the unique power index that is symmetric, efficient, satisfies both the null player property and the transfer property,
see \cite{dubey1975uniqueness}, this result does not transfer to general $(j,k)$ simple games.

\begin{proposition}
  \label{prop_counter_example}
  When $j\geq 3$, there exists some $a\in J^n$ such that $\Phi^{a}\neq \Phi$.
\end{proposition}
\begin{proof}
  Consider the $(j, k)$ simple game $u^b$ with point-veto $b=(1, j-1, 0,\cdots, 0)\in J^n$ and let $a=(j-2,j-2,\cdots,j-2)\in J^n$.
  From Equation~(\ref{eq_parametric_index}) we conclude $\Phi^{a} (u^{b})=\left(0,\, 1,\, 0,\,\cdots, \,0\right)$. Using Equation~(\ref{eq_roll_call_jk_simple_games_simplified}) we
  easily compute $\Phi (u^{b})=\left(\displaystyle \frac{1}{j},\, \displaystyle\frac{j-1}{j}, \,0,\,\cdots,\, 0\right)\neq \Phi^{a}$.
\end{proof}

We remark that the condition $j\ge 3$ is necessary in Proposition~\ref{prop_counter_example}, since for $(2,2)$ simple games the roll-call interpretation
of Mann and Shapley, see \cite{MannShapley}, for the Shapley-Shubik index for simple games yields  $\Phi^{\mathbf{0}}=\Phi^{\mathbf{1}}= \Phi$.

\section{The average game of a $(j,k)$ simple game}
\label{sec_average_game}

Equation~(\ref{eq_roll_call_jk_simple_games_simplified}) in Lemma~\ref{lemma_simplified_formula} has the important consequence that $\Phi(v)$ equals the
Shapley value of the TU game $C(v,\cdot)$, where a TU game is a mapping $v\colon 2^N\to\mathbb{R}$ with $v(\emptyset)=0$. To this end we introduce
an operator that associates each $(j,k)$ simple game $v$ with a TU game $\widetilde{v}$ as follows.

\begin{definition}
  \label{def_aver_tu_j_k}
  Let $v\in\mathcal{U}_n$ be an arbitrary $(j,k)$ simple game. The \emph{average game}, denoted by $\widetilde{v}$, associated to $v$  is defined by
  \begin{equation}
    \label{eq_aver_tu_j_k}
    \widetilde{v}(S)=\displaystyle\frac{1}{j^{n}(k-1)}\sum_{x\in J^{n}}{\left[v(\mathbf{(j-1)}_S\,,x_{-S})-v(\mathbf{0}_S\,,x_{-S})\right]}
  \end{equation}
  for all $ S\subseteq N$.
\end{definition}

With that notation our above remark reads:

 \begin{theorem}
  \label{thm_aver_form_SSI}
  For every $(j,k)$ simple game $v$ the vector $\Phi(v)$ equals the Shapley value of $\widetilde{v}$.
\end{theorem}

For the $(j,k)$ simple game $v$ from Example~\ref{ex_j_k_simple_game} the average simple game is given by
$$
  \tilde{v}(\emptyset)=0, \tilde{v}(\{1\})=\frac{1}{2}, \tilde{v}(\{2\})=\frac{2}{3},\text{ and }\tilde{v}(N)=1.
$$

Before giving some properties of the average game operator we note that two distinct $(j,k)$ simple games may have the same average game, as illustrated in the following example.
\begin{example}
  \label{ex_different_j_k_same_average}
  Consider the $(j, k)$ simple games $u,v\in\mathcal{U}_n$ defined by
  \begin{itemize}
    \item $u(x)=k-1$ if $x=\mathbf{1}$ and $u(x)=0$ otherwise;
    \item $v(x)=k-1$ if $x\neq \mathbf{0}$ and $v(x)=0$ otherwise
  \end{itemize}
  for all $x\in J^n$. Obviously, $u\neq v$. A simple calculation, using Equation~(\ref{eq_aver_tu_j_k}),  gives $\tilde{u}(S)=\tilde{v}(S)=\displaystyle \frac{1}{j^{n-s}}$
  for all $S\in 2^{N}$.
\end{example}

The average game operator has some nice properties among which are the following:
\begin{proposition}
  \label{prop_average_game}
  Given a $(j,k)$ simple game $v\in\mathcal{U}_n$,
  \begin{itemize}
    \item[(a)] $\widetilde{v}$ is a TU game on $N$ that is $\left[0,1\right]$-valued and monotone;
    \item[(b)] any null player in $v$ is a null player in  $\widetilde{v}$;
    \item[(c)] any two equivalent players in $v$ are equivalent in $\widetilde{v}$;
    \item[(d)] if $v=\sum_{t=1}^p\alpha_tv_t$ is a convex combination for some $v_1,\dots,v_p\in \mathcal{U}_n$, then $\widetilde{v}=\sum_{t=1}^p\alpha_t\widetilde{v_t}$.
  \end{itemize}
\end{proposition}
\begin{proof}
Let $v\in \mathcal{U}_n$.%, it is clear that $\tilde{v}$ is a TU game on $N$ such that $0\leq \tilde{v}(S)\leq 1$ for all $S\subseteq N$.
All mentioned properties of $\widetilde{v}$ are more or less transfered from the corresponding properties of $v$ via Equation~(\ref{eq_aver_tu_j_k}).
More precisely:
\begin{itemize}
  \item[(a)]  Note that $\widetilde{v}(\emptyset)=\frac{1}{j^{n}(k-1)}\sum_{x\in J^{n}}{\left[v(x)-v(x)\right]}=0$ and
              $\widetilde{v}(N)=\frac{1}{j^{n}(k-1)}\sum_{x\in J^{n}}{[v(\mathbf{(j-1)})-v(\mathbf{0})]}=1$. Since $v$ is monotone and
              $\mathbf{0}\le x\le\mathbf{(j-1)}$ for all $x\in J^n$, we have $v(\mathbf{(j-1)}_S,x_{-S})\le v(\mathbf{(j-1)}_S,x_{-T})$ and
              $v(\mathbf{0}_S,x_{-S})\ge v(\mathbf{0}_S,x_{-T})$ for all $\emptyset \subseteq S\subseteq T\subseteq N$. Thus, we can
              conclude $0\le \widetilde{v}(S)\le \widetilde{v}(T)\le 1$ from Equation~(\ref{eq_aver_tu_j_k}).
  \item[(b)]  Let $i\in N$ be a null player in $v$ and $S\subseteq N\backslash\{i\}$. Since
              $v(\mathbf{(j-1)}_{S\cup\{i\}},x_{-(S\cup\{i\})})=v(\mathbf{(j-1)}_S\,,x_{-S})$ and
              $v(\mathbf{0}_{S\cup\{i\}},x_{-(S\cup\{i\})})=v(\mathbf{0}_S\,,x_{-S})$, we have that $\widetilde{v}(S\cup\{i\})=
              \widetilde{v}(S)$, i.e., player $i$ is a null player in $\widetilde{v}$.
  \item[(c)]  Let $i,h\in N$ be two equivalent players in $v$, $S\subseteq N\backslash\{i, h\}$, and $\pi_{ih}\in\mathcal{S}_n$ the transposition that interchanges
              $i$ and $h$. Since $v(\mathbf{(j-1)} _{S\cup\{i\}},x_{-(S\cup\{i\})})=v(\mathbf{(j-1)}_{S\cup\{h\}}\,,\left(\pi_{ih}x\right)_{-S\cup\{h\}})$ and
              $v(\mathbf{0} _{S\cup\{i\}},x_{-(S\cup\{i\})})=v(\mathbf{0}_{S\cup\{h\}}\,,\left(\pi_{ih}x\right)_{-S\cup\{h\}})$, we have $\widetilde{v}(S\cup\{i\})=
              \widetilde{v}(S\cup\{h\})$, i.e., players $i$ and $h$ are equivalent in $\widetilde{v}$.
  \item[(d)] Now suppose that $v=\sum_{t=1}^{p}{\alpha_t v_t}$ is a convex combination for some $v_1, v_2,\cdots, v_p \in \mathcal{U}_n$. Since
             $v(\mathbf{(j-1)}_S\,,x_{-S})=\sum_{t=1}^{p}{\alpha_t v_t(\mathbf{(j-1)}_S\,,x_{-S})}$ and $v(\mathbf{0}_S\,,x_{-S})=\sum_{t=1}^{p}{\alpha_t v_t(\mathbf{0}_S\,,x_{-S})}$,
             Equation~(\ref{eq_aver_tu_j_k}) gives $\widetilde{v}(S)=\sum_{t=1}^{p}{\alpha_t \widetilde{v_t}}(S)$ for all $\emptyset\subseteq S\subseteq N$.
\end{itemize}
\end{proof}

The operator that associates each $(j,k)$ simple game $v$ with its average game $\widetilde{v}$ can be seen as a coalitional representation of $(j,k)$ simple games. Moreover,
Proposition \ref{prop_average_game} suggests that this representation preserves some properties of the initial game. The average game of a $(j,k)$ simple game with a point-veto
is provided by:

\begin{proposition}
  \label{prop_aver_veto_j_k}
  Given $a\in J^n\backslash \{\mathbf{0}\}$, the average game $\widetilde{u^a}$ satisfies for every coalition $S\neq N$
  \[
\widetilde{u^{a}}\left( S\right) =\displaystyle\left\{
\begin{array}{ccc}
\displaystyle\prod_{i\in N\backslash S} {\displaystyle \left(\frac{j-a_i}{
j}\right)} & \text{if} & S\cap N^{a}\neq \emptyset  \\
0 & \text{if} & S\cap N^{a}= \emptyset
\end{array}%
\right.
\]%
\end{proposition}
\begin{proof}
  Let $a\in J^n\backslash \{\mathbf{0}\}$ and $\emptyset\subsetneq S\subsetneq N$.

  First  suppose that $S\cap N^{a}= \emptyset$. Then, for all $x\in J^n$ we have $a\leq (\mathbf{(j-1)}_S,\, x_{-S})$ iff $a\leq (\mathbf{0}_S,\, x_{-S})$. Thus,
  $u^a((\mathbf{(j-1)}_S,\, x_{-S}))=u^a((\mathbf{0}_S,\, x_{-S}))$. It then follows from (\ref{eq_aver_tu_j_k}) that $\widetilde{u^a}(S)=0$.

  Now suppose that $S\cap N^{a}\neq \emptyset$. Then, for all $x\in J^n$ we have $a\nleq(\mathbf{0}_S,\, x_{-S})$. Thus, $u^a((\mathbf{0}_S,\, x_{-S}))=0$. Note that
  $a\leq (\mathbf{(j-1)}_S,\, x_{-S})$ iff $a_{-S}\leq x_{-S}$. Hence,
  \begin{eqnarray*}
  \widetilde{u^a}(S) &=& \displaystyle\frac{1}{j^{n}(k-1)}\sum_{x\in J^{n}}{u^a(\mathbf{(j-1)}_S\,,x_{-S})}
  =  \displaystyle\frac{1}{j^{n-s}(k-1)}\sum_{x_{-S}\in J^{-S}}{u^a(\mathbf{(j-1)}_S\,,x_{-S})}  \\
  &=& \displaystyle\frac{1}{j^{n-s}(k-1)}\sum_{x_{-S}\in J^{-S}\, \wedge \,a_{-S}\leq x_{-S}}{u^a(\mathbf{(j-1)}_S\,,x_{-S})}\\
  &=&\displaystyle\frac{1}{k-1}\cdot (k-1)\displaystyle\frac{|\{x_{-S}\in J^{-S},\,  a_{-S}\leq x_{-S}\}|}{j^{n-s}}
  = \displaystyle\prod_{i\in N\backslash S}{\displaystyle\left(\frac{j-a_i}{j}\right)}.
  \end{eqnarray*}
  %%In both cases, $\widetilde{u^a}(S)$ is completely determined.
\end{proof}

It may be interesting to check whether each $(j,k)$ simple game may be decomposed as a combination of $(j,k)$ simple game with a point-veto of the  form  $a\in \{0, j-1\}^n$. The
response is affirmative when one considers combinations between average games. Before we prove this, recall that the average game associated with each $(j,k)$ simple game is
a TU game on $N$. The set of all TU games on $N$ is vector space and a famous basis consists in all unanimity games $(\gamma_S)_{S\in 2^N}$, where $\gamma_S(T)=1$ if $S\subseteq T$
and $\gamma_S(T)=0$ otherwise.\footnote{The definition of unanimity games has already been given in the second paragraph of Section~\ref{sec_preliminaries}.}

In Definition~\ref{def_point_veto} we have introduced the notation $w^S=u^{a}$ for a coalition $S\in 2^N$, where $a\in J^n$ is specified by $a_i=j-1$ if $i\in S$ and  $a_i=0$ otherwise.

\begin{proposition}
  \label{prop_ws_tilde_gamma_s}
  For every coalition $C\in 2^N$, there exists a collection of real numbers $(y_S)_{S\in 2^C}$ such that
  $$\widetilde{w^C}=\sum_{S\in {2^C}}y_S\gamma_S.$$
\end{proposition}
\begin{proof}
  Note that $\widetilde{w^C}$ is a TU game on $N$. Therefore, for some real numbers $(y_S)_{S\in 2^N}$ we have
  \begin{equation}
  \label{eq_ws_tilde_gamma_s_1}
  \widetilde{w^C}=\sum_{S\in {2^N}}y_S\gamma_S.
  \end{equation}
  This proves the result for $C=N$. Now, suppose that $C\neq N$. Consider $\mathcal{ E}_k=\{T\in 2^{N}, T\backslash C\neq \emptyset\quad\text{and}\quad |T|=k\}$ for $1\leq k\leq n$.
  We prove by induction on $k$ that $y_{T}=0$ for all coalitions $T\in \mathcal{E}_k$. More formally, consider the assertion $\mathcal{P}(k): \,\text{for all}\, T\in \mathcal{ E}_k$,
  we have $y_T=0$.

  First assume that, $k=1$. Let $T\in \mathcal{E}_k$, then there exists $i\in N\backslash C$ such that $T=\{i\}$. Since player~$i$ is not contained in $C$,
  Proposition~\ref{prop_null_and_equivalent_players} and Proposition~\ref{prop_average_game} yield that $i$ is a null player in $\widetilde{w^C}$, so that
  $\widetilde{w^C}(T)=0$. Since Equation~$(\ref{eq_ws_tilde_gamma_s_1})$ gives $\widetilde{w^C}(T)=\sum_{S\in {2^{T}}}{y_S}=y_T$, we have $y_T=0$. Therefore
  $\mathcal{P}(1)$ holds. Now consider $2\leq k\leq n$ suppose and $\mathcal{P}(l)$ holds for all $1\leq l<k$. Let $T\in \mathcal{ E}_k$, then there exists
  $i\in N\backslash C$ such that $T=K\cup\{i\}$, $i\notin K \neq\emptyset $. Since  $i$ is a null player in $\widetilde{w^C}$, we have $\widetilde{w^C}(T)-\widetilde{w^C}(K)=0$.
  Using Equation~$(\ref{eq_ws_tilde_gamma_s_1})$ we compute:
  \begin{align*}
   \widetilde{w^C}(T)-\widetilde{w^C}(T)=& \sum_{S\in {2^T}}{y_S}-\sum_{S\in {2^K}}{y_S}
  =  y_T+\sum_{i\in S\varsubsetneq T}{y_S}
  =  y_T
  \end{align*}
  using $S\backslash C\neq \emptyset$ and $1\leq |S|<|T|=k$. Thus, we have $y_T=0$, which proves our claim.
\end{proof}

\begin{proposition}
  \label{prop_decompo_of_u_a_tilde}
  For every $(j,k)$ simple game $u\in\mathcal{U}_n$, there exists a collection of real numbers $(x_S)_{S\in 2^N}$ such that
  \begin{equation}\label{eq u a tilde w s}
  \widetilde{u}=\sum_{S\in 2^N}x_S\widetilde{w^S}.
  \end{equation}
\end{proposition}
\begin{proof}
The result is straightforward when $j=2$ since $J$ reduces to $J=\{0,1\}.$ In the rest of the proof, we assume that $j\geq 3.$  Note that all TU games on $N$ can be written
as a linear combination of unanimity games $(\gamma_S)_{S\in 2^N}$. It is then sufficient to only prove that each TU game $\gamma_C$ for $C\in 2^N$ is a linear combination
of the TU games $(\widetilde{w^S})_{S\in 2^C}.$ The proof is by induction on $1\leq k=|C| \leq n$. More precisely,  we prove the assertion $\mathcal{A}(k)$ that for all
$C\in 2^N$ such that $|C|\leq k$, there exists a collection $(z_S)_{S\in 2^C}$ such that
\begin{equation}\label{eq gamma and tilde w s}
\gamma_C=\sum_{S\in {2^C}}z_S\widetilde{w^S}.
\end{equation}

First assume that $k=1$.  Using Proposition~\ref{prop_aver_veto_j_k},  it can be easily checked that we have $\gamma_{\{i\}}=\widetilde{w^{\{i\}}}$ for all $i\in N$.
Therefore $\mathcal{A}(1)$ holds.  Now, consider a coalition $C$ such that $|C|=k\in \{2,\dots,n\}$ and assume that $\mathcal{A}(l)$ holds for all $l$ such that $1\leq l<k$.
By Proposition~\ref{prop_ws_tilde_gamma_s}, there exists some real numbers $(\alpha_S)_{S\in 2^C}$ and $(\beta_S)_{S\in 2^C\backslash \{C\}}$ such that
$$\widetilde{w^C}=\sum_{S\in {2^C}}\alpha_S\gamma_S=\alpha_C\gamma_C+\sum_{S\in {2^C\backslash \{C\}}}\alpha_S\gamma_S=\alpha_C\gamma_C+\sum_{S\in {2^C\backslash \{C\}}}\beta_S\widetilde{w^S}.$$
where the last equality holds by the induction hypothesis.
Moreover, $\alpha_C$ can be determined using Proposition~\ref{prop_aver_veto_j_k} for $c=|C|$ by:
$$\alpha_C=\sum_{S\in {2^C}}(-1)^{|C\backslash S|}\widetilde{w^C}(S)=\displaystyle\sum_{s=1}^c{(-1)^{c-s}\binom{c}{s}\left(\frac{j-1}{j}\right)^{c-s}}=\frac{1-(1-j)^{c}}{j^{c}}\neq 0\text{ since } j-1\geq 2.$$

Therefore we get
$$\gamma_C=\sum_{S\in {2^C}}z_S\widetilde{w^S}$$
where for all $S\in 2^C$,  $z_S=-\frac{1}{\alpha_C}$ if $S=C$ and $z_S=-\frac{\beta_S}{\alpha_C}$ otherwise. This gives $\mathcal{A}(k)$. In summary, each $\gamma_S,S\in 2^N$ is a
linear combination of $\widetilde{w^C}, C\in 2^N$. Thus, the proof is completed since $\widetilde{u}$ is a linear combination of $\gamma_S,S\in 2^N$.
\end{proof}

Before we continue, note that by Equation~(\ref{eq gamma and tilde w s}), for $C\in 2^N$ each TU game $\gamma_C$ is a linear combination of the TU games  
$\left(\widetilde{w^S}\right)_{S\in 2^N}.$ Since  $\left(\gamma_S\right)_{S\in 2^N}$ is a basis of the vector space of all TU games on $N$, it follows that 
$\left(\widetilde{w^S}\right)_{S\in 2^N}$ is also a basis of the vector space of all TU games on $N$.

\section{A characterization of the Shapley-Shubik index for $(j,k)$ simple games}
\label{sec_axiomatization_j_k_1}

As shown in Proposition~\ref{prop_parametric_power_index} the axioms of Definition~\ref{def_power_index_properties} are not sufficient to uniquely characterize the power
index $\Phi$ for the class of $(j,k)$ simple games. Therefore we introduce an additional axiom.

\begin{definition}
  \label{def_averagely_convex}
  A power index $F$ for $(j,k)$ simple games is \emph{averagely convex} (AC) if we always have
  \begin{equation}\label{eq_axiom_AC_jk_1}
    \sum_{t=1}^p\alpha_tF(u_t)=\sum_{t=1}^q\beta_tF(v_t)
  \end{equation}
  whenever
  \begin{equation}\label{eq_axiom_AC_jk_2}
    \sum_{t=1}^p\alpha_t\widetilde{u_t}=\sum_{t=1}^q\beta_t\widetilde{v_t},
  \end{equation}
  where $u_1,u_2,\dots,u_p,v_1,v_2,\dots,v_q\in  \mathcal{U}_n$ and $(\alpha_t)_{1\leq t \leq p}$, $(\beta_t)_{1\leq t \leq q}$ are non-negative numbers that sum to $1$ each.
\end{definition}

 One may motivate the axiom (AC) as follows. In a game, the \emph{a priori strength} of a coalition, given the profile of the other individuals, is the difference between
the outputs observed when all of her members respectively give each her \emph{maximum support} and her \emph{minimum support}. The \emph{average strength} game associates
each coalition with her expected strength when the profile of other individuals uniformly varies. Average convexity for power indices is the requirement that whenever the
average game of a game is a convex combination of the average games of two other games, then the same convex combination still applies for the power distributions.

We remark that the axiom of Average Convexity is much stronger than the axiom of Convexity. A minor technical point is that $\sum_{t=1}^p\alpha_t u_t$ as well as
$\sum_{t=1}^q\beta_t v_t$ do not need to be $(j,k)$ simple games. However, the more important issue is that
$$
  \widetilde{\sum_{t=1}^p\alpha_tu_t}\overset{\,\,\text{Proposition~\ref{prop_average_game}.(d)}\,\,}{=}\sum_{t=1}^p\alpha_t\widetilde{u_t}=\sum_{t=1}^q\beta_t\widetilde{v_t}
  \overset{\,\,\text{Proposition~\ref{prop_average_game}.(d)}\,\,}{=}\widetilde{\sum_{t=1}^q\beta_t v_t},
$$
i.e., Equation~(\ref{eq_axiom_AC_jk_2}), is far less restrictive than
$$
  \sum_{t=1}^p\alpha_t u_t=\sum_{t=1}^q\beta_t v_t
$$
since two different $(j,k)$ simple games may have the same average game, see Example~\ref{ex_different_j_k_same_average}.
Further evidence is given by the fact that the parametric power indices $\Phi^a$, defined in Equation~(\ref{eq_parametric_index}),
do not all satisfy (AC).

\begin{proposition}
  \label{prop_counter_example2}
  When $j\geq 3$, there exists some $a\in J^n$ such that $\Phi^{a}$ does not satisfy (AC).
\end{proposition}
\begin{proof}
  As in the proof of Proposition~\ref{prop_counter_example}, consider the $(j, k)$ simple game with point-veto $b=(1, j-1, 0,\cdots, 0)\in J^n$ and let
  $a=(j-2,j-2,\cdots,j-2)$. It can be easily checked that, for all subsets $T\subseteq N$ we have
  $$\widetilde{u^{b}}(T)=\left\{\begin{array}{ccl}
                                  1 & \text{if} & 1, 2\in T \\
                                  (j-1)/j & \text{if} & 2\in T\subseteq N\backslash\{1\}  \\
                                  1/j & \text{if} & 1\in T\subseteq N\backslash\{2\} \\
                                  0 & \text{if} & T\subseteq N\backslash\{1, 2\}
                                \end{array}\right.$$
  and that
                                \begin{equation}\label{eq-decomposition-example}
                                 \widetilde{u^b}=\displaystyle\frac{1}{j}\cdot \widetilde{w^{\{1\}}}+\displaystyle\frac{j-1}{j}\cdot \widetilde{ w^{\{2\}}}
                                \end{equation}
  holds. Since $\Phi^a$ satisfies (NP), (E), (S) we can easily compute $\Phi^a\left(w^{\{1\}}\right)=(1,\,0,\, \cdots,\, 0)$ and $\Phi^a\left(w^{\{2\}}\right)=(0,\,1,\,0,\, \cdots,\, 0)$.
  Therefore,
  \begin{equation}\label{Eq-AC-for-parametric-power}
  \displaystyle\frac{1}{j}\cdot \Phi^a\left(w^{\{1\}}\right)+\displaystyle\frac{j-1}{j}\cdot \Phi^a\left( w^{\{2\}}\right)=\left(\displaystyle \frac{1}{j},\, \displaystyle\frac{j-1}{j}, \,0,\,\cdots,\, 0\right).
  \end{equation}
 Using (\ref{eq_parametric_index}), one gets $\Phi^{a} (u^{b})=\left(0,\, 1,\, 0,\,\cdots, \,0\right)$. It then follows from equations (\ref{eq-decomposition-example}) and
 (\ref{Eq-AC-for-parametric-power}) that $\Phi^a$ does not satisfy (AC).
\end{proof}

As a preliminary step to our characterization result in Theorem~\ref{thm_charact_SSI_j_k} we state:
\begin{lemma}
  \label{lemma_E_S_NP}
  If a power index $F$ for the class $\mathcal{U}_n$ of $(j,k)$ simple games satisfies (E), (S), and (NP), then we have
  $F(w^C) =\Phi(w^C)$ for all $ C\in 2^N$.
\end{lemma}
\begin{proof}
  Let $F$ be a power index on $\mathcal{U}_n$ that satisfies (E), (S), (NP) and let $C\in 2^N$ be arbitrary.

  According to Proposition~\ref{prop_null_and_equivalent_players}, all players $i,j\in C$ are equivalent in $w^C$ and those outside of $C$ are null
  players in the game $w^C$. Since both $F$ and $\Phi$ satisfy (E), (S), and (NP), we have $F_i(w^C) =\Phi_i(w^C)=\displaystyle\frac{1}{|C|}$ if $i\in C$ and
  $F_i(w^C) =\Phi_i(w^C)=0$ otherwise. %% It is then clear that $F(w^C) =\Phi(w^C)$.
\end{proof}

\begin{theorem}
  \label{thm_charact_SSI_j_k}
  A power index $F$ for the class $\mathcal{U}_n$ of $(j,k)$ simple games satisfies (E), (S), (NP), and (AC) if and only if $F=\Phi$.
\end{theorem}
\begin{proof}
  \underline{\emph{Necessity}}: As shown in Proposition~\ref{prop_power_index_properties}, $\Phi$ satisfies (E), (S), and (NP). For (AC) the proof follows from
  Theorem~\ref{thm_aver_form_SSI} since  the average game operator is linear by Proposition \ref{prop_average_game}.

  \underline{\emph{Sufficiency}}: Consider a power index $F$ for $(j,k)$ simple games that satisfies (E), (S), (NP), and (AC). Next, consider an arbitrary $(j,k)$
  simple game $u\in\mathcal{U}_n$. By Proposition~\ref{prop_decompo_of_u_a_tilde}, there exists a collection of real numbers $(x_S)_{S\in 2^N}$ such that
\begin{equation}\label{eq u a tilde w s pos neg}
\widetilde{u}=\sum_{S\in 2^N}x_S\widetilde{w^S}=\sum_{S\in E_1}x_S\widetilde{w^S}+\sum_{S\in E_2}x_S\widetilde{w^S},
   \end{equation}
where $E_1=\{S\in 2^N:x_S>0\}$ and $E_2=\{S\in 2^N:x_S<0\}$. Note that $E_1\neq \emptyset$ since $\widetilde{u}(N)=1$. As an abbreviation we set
\begin{equation}\label{eq u a tilde w s pos neg pos}
\varpi=\sum_{S\in E_1}x_S\widetilde{w^S}(N)=\sum_{S\in E_1}x_S>0.
   \end{equation}
It follows that
\begin{equation}\label{eq u a tilde w s pos neg conv}
\frac{1}{\varpi}\widetilde{u}+\sum_{S\in E_2}\frac{-x_S}{\varpi}~\widetilde{w^S}=\sum_{S\in E_1}\frac{x_S}{\varpi}~\widetilde{w^S}.
   \end{equation}
 Since $(\ref{eq u a tilde w s pos neg conv})$ is an equality among two convex combinations, axiom (AC) yields
\begin{equation*}
\frac{1}{\varpi}F(u)+\sum_{S\in E_2}\frac{-x_S}{\varpi}~F(w^S)=\sum_{S\in E_1}\frac{x_S}{\varpi}~F(w^S).
   \end{equation*}
Therefore by Lemma \ref{lemma_E_S_NP},
\begin{equation}\label{eq u a tilde w s pos neg 3}
\frac{1}{\varpi}F(u)+\sum_{S\in E_2}\frac{-x_S}{\varpi}~\Phi(w^S)=\sum_{S\in E_1}\frac{x_S}{\varpi}~\Phi(w^S).
   \end{equation}
Since $\Phi$ also satisfies (AC), we obtain
\begin{equation}\label{eq u a tilde w s pos neg 4}
\frac{1}{\varpi}F(u)+\sum_{S\in E_2}\frac{-x_S}{\varpi}~\Phi(w^S)=\frac{1}{\varpi}\Phi(u)+\sum_{S\in E_2}\frac{-x_S}{\varpi}~\Phi(w^S),
   \end{equation}
so that $F(u)=\Phi(u)$.
\end{proof}
 
 \begin{proposition}
 For $j\ge 3$, the four axioms in Theorem \ref{thm_charact_SSI_j_k}  are independent.
\end{proposition}
\begin{proof} 
  For each of the four axioms in Theorem \ref{thm_charact_SSI_j_k}, we provide a power index on $\mathcal{U}_n$ that meets the three other axioms but not the chosen one.
\begin{itemize}
  \item The power index $2\cdot \Phi$ satisfies (NP), (S), and (AC) but not (E).
  \item Denote by $\operatorname{ED}$ the \emph{equal division} power index which assigns $\tfrac{1}{n}$ to
        each player for every $(j, k)$ simple game $v$. Then, the power index $\frac{1}{2}\cdot\Phi +\frac{1}{2}\cdot\operatorname{ED}$
        satisfies (E), (S) and (AC), but not (NP).
  \item In Proposition~\ref{prop_parametric_power_index} we have constructed a parametric series of power indices that satisfiy  (E), (S), and (NP). 
        For $j\ge 3$, at least one example does not satisfy (AC), see Proposition~\ref{prop_counter_example2}.
  \item Recall that $\left(\widetilde{w^S}\right)_{S\in 2^N}$ is a basis of the vector space of all TU games on $N$. Thus given a $(j,k)$-simple game $u$, there exists  
        a unique collection of real numbers $\left(x_S^u\right)_{S\in 2^N}$ such that 
        \begin{equation}\label{eq u a tilde w s bis 1}
          \widetilde{u}=\sum_{S\in 2^N}x_S^u \widetilde{w^S}.
        \end{equation}
        Consider some $i_0\in N$ and set 
  \begin{equation}\label{eq u a tilde w s bis 2}
 F( u)=\sum_{S\in 2^N}x_S^u\cdot F\!\left(w^S\right).
  \end{equation}
  For each $S\in 2^N\backslash \{N\}$ we set $F_{i}\!\left(w^{S}\right)=\Phi\!\left(w^S\right)$. For $S=N$ we set $F_i\!\left(w^N\right)=\displaystyle\frac{2}{n+1}$ if $i=i_0$ and 
  $F_i\!\left(w^{N}\right)=\displaystyle\frac{1}{n+1}$ otherwise. We can easily check that $F$ satisfies (E), (NP), (AC), but not (S).
      \end{itemize}
This proves that the four axioms in Theorem \ref{thm_charact_SSI_j_k}  are independent.
\end{proof}
      
\section{Axiomatization of the Shapley-Shubik index for interval simple games}
\label{sec_axiomatization_interval_simple_games}

Similar as for $(j,k)$ simple games a Shapley-Shubik like index for interval simple games can be constructed from the idea of the roll-call model.

\begin{definition}(cf.~\cite[Definition 6.2]{kurz2014measuring})\\
  Let $v$ be an interval simple game with player set $N$ and $i\in N$ an arbitrary player. We set
  \begin{eqnarray}
  \Psi_i(v) &=& \frac{1}{n!}\sum_{\pi\in\mathcal{S}_n}\int_0^1\dots\int_0^1 \left[v(x_{\pi_{<i}},\mathbf{1}_{\pi_{\ge i}})-v(x_{\pi_{<i}},\mathbf{0}_{\pi_{\ge i}})\right]\notag\\
  &&-\left[v(x_{\pi_{\le i}},\mathbf{1}_{\pi_{> i}})-v(x_{\pi_{\le i}},\mathbf{0}_{\pi_{> i}})\right]\operatorname{d}x_1\dots \operatorname{d}x_n.\label{eq_def_ssi_interval_sg}
  \end{eqnarray}\label{def_ssi_interval_sg}
\end{definition}

In this section, we give a similar axiomatization for $\Psi$ (for interval simple games) as we did for $(j,k)$ simple games and $\Phi$. By a power index for
interval simple games we understand a mapping from the set of interval simple games for $n$ players to $\mathbb{R}^n$. Replacing $J$ by $I=[0,1]$ in
Definition~\ref{def_power_index_properties}, allows us to directly transfer the properties of power indices for $(j,k)$ simple games to the present situation.
Also Proposition~\ref{prop_power_index_properties} is valid for interval simple games and $\Psi$. More precisely, $\Psi$ satisfies
(P), (A), (S), (E), (NP), and (T), see \cite[Lemma 6.1]{kurz2018importance}. The proof for (C) and (L) goes along the same lines as the proof of
Proposition~\ref{prop_power_index_properties}. Also the generalization of the power index to a parametric class can be done just as the one for
$(j,k)$ simple games in Equation~(\ref{eq_parametric_index}).

\begin{proposition}
  \label{prop_generalized_measure_point}
  For every $\alpha\in[0,1]$ the mapping $\Psi^a$, where $a=(\alpha,\dots,\alpha)\in[0,1]^n$, defined by
  $$
  \Psi^a_i(v) = \frac{1}{n!}\sum_{\pi\in\mathcal{S}_n}\Big(\left[v(a_{\pi_{<i}},\mathbf{1}_{\pi_{\ge i}})-v(a_{\pi_{<i}},\mathbf{0}_{\pi_{\ge i}})\right]
  -\left[v(a_{\pi_{\le i}},\mathbf{1}_{\pi_{> i}})-v(a_{\pi_{\le i}},\mathbf{0}_{\pi_{> i}})\right]\Big)
  $$
  for all $i\in N$, satisfies (P), (A), (S), (E), (NP), (T), (C), and (L).
\end{proposition}

Again, there exist vectors $a\in[0,1]^n$ and interval simple games $v$ with $\Psi^a(v)\neq \Psi(v)$. Also the simplified formula
for $\Phi$ for $(j,k)$ simple games in Lemma~\ref{lemma_simplified_formula} can be mimicked for interval simple games and $\Psi$, see \cite{kurz2019axiomatization}.

\begin{proposition}
  \label{prop_formula_ssi_interval_sg}
  For every interval simple game $v$ with player set $N$ and every player $i\in N$ we have
  \begin{equation}
  \label{eq_ssi_interval_sg}
  \Psi_i(v)=\sum_{i\in S\subseteq N} \frac{(s-1)!(n-s)!}{n!}\cdot\left[C(v,S)-C(v,S\backslash\{i\})\right],
  \end{equation}
  where $C(v,T)=\int_{[0,1]^n} v(\mathbf{1}_T,x_{-T})-v(\mathbf{0}_T,x_{-T})\,\operatorname{d}x$ for all $T\subseteq N$.
\end{proposition}

This triggers:

\begin{definition} \label{Def aver-game}
  Let $v$ be an interval simple game on $N$. The \emph{average game} associated with $v$ and denoted by $\widehat{v}$ is defined via
\begin{equation} \label{Eq aver-TU interval sg}
\forall S\subseteq N,\, \widehat{v}(S)=\displaystyle\int_{I^n}{[v(\mathbf{1}_S\,,x_{-S})-v(\mathbf{0}_S\,,x_{-S})]dx}.
\end{equation}
\end{definition}

\begin{theorem} \label{Prop-Def-SSI-csg}
  For all every interval simple game $v$ on $N$ and for all $i\in N$,
  \begin{equation}\label{Eq-SSI}
   \Psi _i(v)=\sum_{i\in S\subseteq N}{\displaystyle\frac{(s-1)!(n-s)!}{n!}[\widehat{v}(S)-\widehat{v}(S\backslash \{i\}]}
  \end{equation}
\end{theorem}
In other words, for a given interval simple game $v$ the power distribution $\Psi(v)$ is given by the Shapley value of its average game $\hat{v}$.

As with $(j, k)$ simple games, two distinct interval simple games may have the same average game as illustrated in the following example.
\begin{example}\label{Exp-different jk-same-average}
  Consider the  interval simple games $u$ and $v$ defined on $N$ respectively for all $x\in [0,1]^n$ by : $u(x)=1$ if $x=\mathbf{1}$, and $u(x)=0$ otherwise; $v(x)=1$ if $x\neq \mathbf{0}$, and $v(x)=0$ otherwise. It is clear that, $u\neq v$.
But, Equation~(\ref{Eq aver-TU interval sg}) and  a simple calculation give $\widehat{u}(S)=\widehat{v}(S)=1$ if $S=N$ and $\widehat{u}(S)=\widehat{v}(S)=0$ otherwise.
\end{example}

We can also transfer Proposition~\ref{prop_average_game}, i.e., the average game operator preserves the following nice properties of interval simple games.
\begin{proposition}\label{prop average-game interval}
For all $v\in \mathcal{CSG}_n$, $\widehat{v}$ is a $\left[0,1\right]$-valued TU game on $N$ such that
\begin{itemize}
\item[(a)] $\widehat{v}$ is a TU game on $N$ that is $\left[0,1\right]$-valued and monotone;
\item[(b)] any null voter in $v$ is null player in  $\widehat{v}$;
\item[(c)] any two symmetric voters in $v$ are symmetric players in $\widehat{v}$;
\item[(d)] if $v=\sum_{t=1}^p\alpha_tv_t$ is a convex combination for some $v_1,\dots,v_p\in \mathcal{U}_n$ then $\widehat{v}=\sum_{t=1}^p\alpha_t\widehat{v_t}$.
\end{itemize}
\end{proposition}
\begin{proof}
  Very similar to the one of Proposition  \ref{prop_average_game}.
\end{proof}

From Theorem~\ref{Prop-Def-SSI-csg} we can directly conclude that $\Psi$ also satisfies \emph{Average Convexity} (AC), which is defined as in Definition~\ref{def_averagely_convex}.

For the remaining part of this section we introduce some further notation. For all $x\in I^n$, let $\mathbf{1}_x=\{i\in N,\, x_i=1\}$; and given a coalition $S$, let $C^S$ be the
interval simple game defined for all $x\in I^n$  by $C^S(x)=1$ if $S\subseteq \mathbf{1}_x$ and  $C^S(x)=0$ otherwise.

\begin{proposition} \label{Prop partic-csg} For all $T\in 2^N$,    $\widehat{C^S}=\gamma_S.$
\end{proposition}

\begin{proof}
Consider $S,T\subseteq N$.
       If $S\subseteq T$ then for all $x\in [0, 1]^n$, $S\subseteq T\subseteq \{i\in N, (\mathbf{1}_T\,,x_{-T})_i=1\}$ and $S\cap\{i\in N, (\mathbf{0}_T\,,x_{-T})_i=1\}=\emptyset$. Then by definition of $C^S$, $C^S(\mathbf{1}_T\,,x_{-T})=1$ and $C^S(\mathbf{0}_T\,,x_{-T})=0$. Therefore,
      \begin{align*}
        \widehat{C^{S}}(T)= & \int_{[0, 1[^n}{[C^S(\mathbf{1}_T,,x_{-T})-C^S(\mathbf{0}_T\,,x_{-T})]dx}=1=\gamma_S(T).
      \end{align*}

      Now assume that $S\nsubseteq T$. Let $x\in [0\, 1)^n$. Note that $\{i\in N, (\mathbf{1}_T\,,x_{-T})_i=1\}=T$ and $\{i\in N, (\mathbf{0}_T\,,x_{-T})_i=1\}=\emptyset$. Therefore,   $S\nsubseteq \{i\in N, (\mathbf{1}_T\,,x_{-T})_i=1\}$ and $S\nsubseteq \{i\in N, (\mathbf{0}_T\,,x_{-T})_i=1\}$. By the definition of $C^S$, it follows that $C^S(\mathbf{1}_T\,,x_{-T})=C^S(\mathbf{0}_T\,,x_{-T})=0$. Hence
      \begin{align*}
        \widehat{C^S}(T)= & \int_{[0, 1[^n}{[C^S(\mathbf{1}_T\,,x_{-T})-C^S(\mathbf{0}_T\,,x_{-T})]dx}=0=\gamma_S(T).
      \end{align*}
In both cases $\widehat{C^S}(T)=\gamma_S(T)$ for all $T\in 2^N$; that is $\widehat{C^S}=\gamma_S$.
\end{proof}

\begin{theorem}\label{Theo_Axiomatization-of SSI}
A power index $\Psi'$ for interval simple games satisfies (E), (S), (NP) and (AC) if and only if $\Psi'=\Psi$.
\end{theorem}
\begin{proof}$\,$\\
\noindent
\underline{\emph{Necessity}}: We have already remarked that $\Psi$ satisfies (E), (S), (AC), and (NP).

\medskip

\noindent
\underline{\emph{Sufficiency}}: Let $\Psi'$ be a power index for interval simple games on $N$ that simultaneously satisfies (E), (S), (AC), and (NP). Consider
an interval simple game $u$. Note that $\widehat{u}$ is a TU game by Proposition~\ref{prop average-game interval}. Thus by Proposition~\ref{Prop partic-csg}, there
exists a collection of real numbers $(\alpha_S)_{S\in 2^N}$ such that
      \begin{equation}\label{Eq-decomposition-aver-for-csg}
        \widehat{u}=\sum_{S\in 2^{N}}{\alpha_S\cdot \widehat{C^S}}=\sum_{S\in E_1}{\alpha_S\cdot \widehat{C^S}}+\sum_{S\in E_2}{\alpha_S\cdot \widehat{C^S}}
      \end{equation}
      where $E_1=\{S\in 2^N:\alpha_S>0\}$ and $E_2=\{S\in 2^N:\alpha_S< 0\}$. Moreover,  $E_1\neq \emptyset$ since $\widehat{v}(N)=1$. We set
\begin{equation}\label{eq-sum-positif-coef-in-decompo-aver-csg}
\varpi=\sum_{S\in E_1}\alpha_S\cdot\widehat{C^S}(N)=\sum_{S\in E_1}\alpha_S>0.
   \end{equation}
It follows that
\begin{equation}\label{eq-Aver-convex-hypothesis-for-csg}
\frac{1}{\varpi}\widehat{u}+\sum_{S\in E_2}\frac{-\alpha_S}{\varpi}~\widehat{C^S}=\sum_{S\in E_1}\frac{\alpha_S}{\varpi}~\widehat{C^S}.
   \end{equation}
Since $(\ref{eq-Aver-convex-hypothesis-for-csg})$ is an equality among two convex combinations, then by (AC), we deduce that
\begin{equation}\label{eq u a tilde w s pos neg csg}
\frac{1}{\varpi}\Psi'(u)+\sum_{S\in E_2}\frac{-\alpha_S}{\varpi}~\Psi'(C^S)=\sum_{S\in E_1}\frac{\alpha_S}{\varpi}~\Psi'(C^S).
   \end{equation}

   Note that given $S\in 2^{N}$, all voters in $S$ are equivalent in $C^S$ while all voters outside $S$ are null players in $C^S$. Since $\Psi'$ and $\Psi$ satisfy (E),
   (S), and (NP), it follows that $\Psi'(C^S)=\Psi(C^S)$. Thus,
\begin{equation}\label{eq u a tilde w s pos neg 3}
\frac{1}{\varpi}\Psi'(u)+\sum_{S\in E_2}\frac{-\alpha_S}{\varpi}~\Psi(C^S)=\sum_{S\in E_1}\frac{\alpha_S}{\varpi}~\Psi(C^S).
   \end{equation}
Since $\Psi$ also satisfies (AC), we get
\begin{equation}\label{eq u a tilde w s pos neg 4}
\frac{1}{\varpi}\Psi'(u)+\sum_{S\in E_2}\frac{-\alpha_S}{\varpi}~\Psi(C^S)=\frac{1}{\varpi}\Psi(u)+\sum_{S\in E_2}\frac{-\alpha_S}{\varpi}~\Psi(C^S).
   \end{equation}
Hence $\Psi'(u)=\Psi(u)$, which proves that $\Psi'=\Psi$.
\end{proof}

\begin{proposition}
 The four axioms in Theorem~\ref{Theo_Axiomatization-of SSI} are independent.
\end{proposition}
\begin{proof}$\,$\\
\begin{itemize}
  \item The power index $2\cdot\Psi$ satisfies (NP), (S), (AC), but not (E).
  \item Denote by $\operatorname{ED}$ the \emph{equal division} power index which assigns $\tfrac{1}{n}$ to
        each player for every interval simple game. Then the power index $\tfrac{1}{2}\cdot\Psi + \tfrac{1}{2}\cdot\operatorname{ED}$
        satisfies (E), (S) and (AC), but not (NP).
  \item In Proposition~\ref{prop_generalized_measure_point} (c.f.~\cite[Proposition 4]{kurz2019axiomatization}) we have stated a parametric classes of power indices
        for interval simple games that satisfy (E), (S), and (NP). In \cite{kurz2019axiomatization} it was also proved that there is at least one parameter $\mathbf{a}$ for which 
        the parameterized index $\Psi^{\mathbf{a}}\neq \Psi$.  Thus, by Theorem~\ref{Theo_Axiomatization-of SSI} we can conclude that $\Psi^{\mathbf{a}}$ does not satisfies (AC).  
        (Also Proposition~\ref{prop_counter_example2} for $(j,k)$ simple games can be adjusted easily.)
  \item Note that by Proposition~\ref{Prop partic-csg} the set $\left(\widehat{C^S}\right)_{S\in 2^N}$ is a basis of the vector space of all TU games on $N$. Thus, given an interval 
        simple game $u$, there exists a unique collection of real numbers $\left(y_S^u\right)_{S\in 2^N}$ such that
  \begin{equation}\label{eq u a tilde w s bis 1}
  \widehat{u}=\sum_{S\in 2^N}y_S^u\widehat{C^S}.
  \end{equation}
Consider some $i_0\in N$ and set 
  \begin{equation}\label{eq u a tilde w s bis 2}
 F(u)=\sum_{S\in 2^N}y_S^u\cdot F\!\left(C^S\right).
  \end{equation}
  For each $S\in 2^N\backslash\{N\}$  we set $F_{i}\!\left(C^{S}\right)=\Phi\!\left(C^S\right)$. For $S=N$ we set $F_i\!\left(C^N\right)=\displaystyle\frac{2}{n+1}$ if $i=i_0$ and  
  $F_i\!\left(C^{N}\right)=\displaystyle\frac{1}{n+1}$ otherwise. We can easily check that $F$ satisfies (E), (NP), (AC), but not (S).
\end{itemize}
This proves that the four axioms in Theorem~\ref{Theo_Axiomatization-of SSI} are independent.
\end{proof}
 
\section*{Acknowledgment}
Hilaire Touyem benefits from a financial support of the CETIC (Centre d'Excellence Africain en
Technologies de l'Information et de la Communication) Project of the University of
Yaounde~I.

%% \bibliography{AverageConvexity}
%% \bibdata{AverageConvexity}
%% %\bibliographystyle{amsplain}
%% \bibliographystyle{abbrv}

\end{document}